\documentclass[conference, 10pt]{IEEEtran}
\usepackage{acronym}

\usepackage{amsthm}

\newtheorem{lemma}{Lemma}
\newtheorem{proposition}{Proposition}

\newtheorem{definition}{Definition}

\newtheorem{Rem}{Remark}

\usepackage{enumitem}
\newcounter{problem}
\newcounter{subproblem}[problem]
\newenvironment{problem}{\refstepcounter{problem}{\bfseries Problem~\theproblem}}{}

\usepackage[bookmarks=false]{hyperref}

\usepackage{mathtools}
\usepackage{color}
\usepackage{xcolor}
\usepackage{colortbl}
\definecolor{purple}{RGB}{139, 0, 139}
\usepackage{manfnt}
\usepackage{url}

\usepackage{caption}
\usepackage{subcaption}

\newif\iftodo   % L"a"st \todo-Eintr"age zu ('Baustellen/Unfertiges')
\todotrue
\newif\iftodoshort  % true: Druckt nur den \todo-Marker ohne Kasten aus
\todoshortfalse
\ifCLASSINFOpdf
  \usepackage[pdftex]{graphicx}
   \DeclareGraphicsExtensions{.pdf,.jpeg,.png,.tiff}
\else
   \usepackage[dvips]{graphicx}
   \DeclareGraphicsExtensions{.eps,.ps}
\fi
\usepackage{amssymb,amsmath}
\usepackage[mathscr]{euscript}
\usepackage{bbm}
\usepackage{dsfont}

\DeclareMathOperator*\argmax{arg \, max}		% arg max
\DeclareMathOperator*\maximize{max.}		% arg max
\usepackage{algorithmic}
\usepackage[ruled,vlined,commentsnumbered]{algorithm2e}

\graphicspath{ {./fig/} }
\newcommand{\Rmnum}[1]{\uppercase\expandafter{\romannumeral #1}}
\newcommand{\rmnum}[1]{\lowercase\expandafter{\romannumeral #1}}

\newcommand{\diag}{\mathop{\mathrm{diag}}}

\newcommand{\field}[1]{\mathbb{#1}}

\newcommand{\emenge}[1]{\mathscr{#1}}
\newcommand{\set}[1]{\mathscr{#1}}
\newcommand{\indication}[1]{\mathds{1}_{\{#1\}}}
\newcommand{\operator}[1]{\mathrm{#1}}

\newcommand{\R}{{\field{R}}}
 
\newcommand{\NN}{{\field{N}}} % natural number: do not include 0

\newcommand{\RN}{{\field{R}}_{+}}
\newcommand{\RP}{{\field{R}}_{++}}

\newcommand{\Ns}{{\emenge{N}}}
\newcommand{\Ms}{{\emenge{M}}}
\newcommand{\Ks}{{\emenge{K}}}
\newcommand{\Ts}{{\emenge{T}}}
\newcommand{\Ss}{{\emenge{S}}}
\newcommand{\Ws}{{\emenge{W}}}

\newcommand{\sinr}{\operator{SINR}}

\newcommand{\row}{\operator{row}}

\newcommand{\ma}[1]{\boldsymbol{\mathbf{#1}}}
\newcommand{\ve}[1]{\boldsymbol{\mathbf{#1}}}

\newcommand{\vx}{\ve{x}}

\newcommand{\vc}{\ve{c}}

\newcommand{\vw}{\ve{w}}

\newcommand{\vnu}{\ve{\nu}}

\newcommand{\vp}{\ve{p}}

\newcommand{\vy}{\ve{y}}

\newcommand{\vf}{\ve{f}}

\newcommand{\vd}{\ve{d}}

\newcommand{\mX}{\ma{X}}
\newcommand{\mY}{\ma{Y}}

\newcommand{\mA}{\ma{A}}
\newcommand{\mB}{\ma{B}}

\newcommand{\mV}{\ma{V}}
\newcommand{\mVt}{\tilde{\ma{V}}}

\newcommand{\mH}{\ma{H}}

\newcommand{\mC}{\ma{C}}

\newcommand{\vt}{\tilde{v}}

\usepackage[utf8]{inputenc}

\newcommand{\ul}{(\text{u})}
\newcommand{\dl}{(\text{d})}

\newcommand{\iter}{\text{iter}}

\usepackage{hyperref}
\usepackage{amsthm}
\usepackage{amssymb,amsmath}
\usepackage[mathscr]{euscript}

\newcommand{\cosl}[1]{}
\newcommand{\resl}[1]{}

\usepackage[draft,footnote,nomargin]{fixme}
\newcommand{\fnql}[1]{}
\newcommand{\fnsv}[1]{}

\begin{document}
%
% paper title
\title{Dynamic Uplink/Downlink Resource Management in Flexible Duplex-Enabled Wireless Networks}
\author{
\IEEEauthorblockN{Qi Liao}
\IEEEauthorblockA{Nokia Bell Labs, Stuttgart, Germany\\ 
Email: \url{qi.liao@nokia-bell-labs.com}}
%\IEEEauthorblockA{\IEEEauthorrefmark{2}Huawei Technologies, France, \url{stefan.valentin@huawei.com}}
}

\maketitle
\begin{abstract}
Flexible duplex is proposed to adapt to the channel and traffic asymmetry for future wireless networks \cite{ngmn_whitepaper}. In this paper, we propose two novel algorithms within the flexible duplex framework for joint uplink and downlink resource allocation in multi-cell scenario, named \ac{SAFP} and \ac{RMDI}, based on the awareness of interference coupling among wireless links. Numerical results show significant performance gain over the baseline system with fixed uplink/downlink resource configuration, and over the dynamic \ac{TDD} scheme that independently adapts the configuration to time-varying traffic volume in each cell. The proposed algorithms achieve two-fold increase when compared with the baseline scheme, measured by the worst-case quality of service satisfaction level, under a low level of traffic asymmetry. The gain is more significant when the traffic is highly asymmetric, as it achieves three-fold increase.  
\end{abstract}

\section{Introduction}\label{sec:intro}
Flexible duplex is one of the key technologies in \ac{5G} to optimize the resource utilization depending on  traffic demand \cite{ngmn_whitepaper}. The main objective is to adapt to asymmetric \ac{UL} and \ac{DL} traffic with flexible resource allocation in the joint time-frequency domain, such that the distinction between \ac{TDD} and \ac{FDD} is blurred, or completely removed.

% flexible duplex
Despite the advantage of adaptation to the dynamic traffic asymmetry, the drawback is the newly introduced \ac{ICI} between duplexing mode \ac{DL} and \ac{UL}, hereinafter referred as {\it \ac{IMI}}. The \ac{DL}-to-\ac{UL} interference plays a more important role due to the large difference between \ac{DL} and \ac{UL} transmission power. Many works focus on physical layer design to overcome \ac{IMI}. In \cite{liu2015performance}, special kinds of radio frames with different ratio of \ac{UL}/\ac{DL} are introduced to \ac{FDD}, and heuristic approach is proposed to find the most suitable one solely based on the traffic volume. A few studies target the problem of dynamic \ac{UL}/\ac{DL} resource configuration. In  \cite{el2011optimized}, the authors formulate a utility maximization problem to minimize the per-user difference between \ac{UL} and \ac{DL} rates; while in \cite{el2012stable} the problem is formulated as a two-sided stable matching game to optimize the average utility per user. Both works consider a single cell system where \ac{IMI} does not play a role. However, in a multi-cell system the optimal \ac{UL}/\ac{DL} configuration depends not only on the traffic volume but also the interference coupling between all transmission links. Although very few studies provide solutions within the flexible duplex framework, similar problem exists in dynamic \ac{TDD}. A popular solution is the cell-cluster-specific \ac{UL}/\ac{DL} reconfiguration \cite{shen2012dynamic}, but how to coordinate the clusters for inter-cluster \ac{IMI} mitigation still remains a challenge. 

% flexible dupex
In this paper, we optimize \ac{UL}/\ac{DL} resource configuration in multi-cell scenario, by recasting max-min fairness problem into a fixed point framework. Such framework is widely used for power control \cite{Yates95a,nuzman2007contraction} and load estimation \cite{siomina2012analysis,cavalcante2016elementary} for \ac{UL} or \ac{DL} systems independently. Our previous work \cite{liao2016dynamic} exploits the framework to tackle the joint \ac{UL}/\ac{DL} resource allocation and power control problem within flexible duplex, assuming that \ac{ICI} is simply proportional to the load. This assumption, however, is valid only when each resource unit has the same chance to be allocated to \ac{UL} or \ac{DL}, which may result in high probability of generating \ac{IMI}. We improved the model in this paper. The main contribution is summarized in below.
\begin{itemize}
\item A new interference model is defined, which allows to prioritize the positions of the resources for \ac{UL} and \ac{DL} transmission, to reduce the probability of generating \ac{IMI}.
\item We propose a novel algorithm \ac{SAFP} to find algorithmic solution to optimize \ac{UL}/\ac{DL} resource configuration. Unlike the models in previous works \cite{Yates95a,siomina2012analysis,cavalcante2016elementary}, the new interference model is {\it nonlinear} and {\it nonmonotonic}. %We show that multiple fixed points may exist, and our proposed algorithm \ac{SAFP} can efficiently find a good, if not optimal, solution.  
\item Further we enhance \ac{SAFP} to \ac{RMDI} by detecting sequentially the dominant interferer in the system, and muting the partial resource in neighboring cells to reduce \ac{ICI}. % Heuristic is introduced to the enhanced algorithm \ac{RMDI} to guarantee an optimized utility no less than \ac{SAFP}.  
\item We compare \ac{SAFP} and \ac{RMDI} numerically with two conventional schemes: a) fixed \ac{UL}/\ac{DL} configuration, and b) dynamic \ac{TDD} that adapts \ac{UL}/\ac{DL} configuration solely based on traffic volume, and show a performance gain varying from two to three fold depending on the traffic asymmetry.  
\end{itemize}

The rest of the paper is organized as follows. In Section \ref{sec:model}, the system model is described together with the correspondent notation. The problem statement is given in Section \ref{sec:Prob}. The proposed algorithms \ac{SAFP} and \ac{RMDI} are introduced in Section \ref{sec:Algor_SAFP} and \ref{sec:Algor_RMDI}, respectively. Finally, in Section \ref{sec:Numeric}, the numerical results are presented.
\section{System Model}\label{sec:model}
In this paper, we use the following definitions. The nonnegative and positive orthant in $k$ dimensions are denoted by $\RN^k$ and $\RP^k$, respectively. Let $\vx\leq\vy$ denote the component-wise inequality between two vectors $\vx$ and $\vy$. Let $\diag(\vx)$ denote a diagonal matrix with the elements of $\vx$ on the main diagonal. For a function $\vf:\R^k\to\R^k$, $\vf^n$ denotes the $n$-fold composition so that $\vf^n=\vf\circ\vf^{n-1}$. The cardinality of set $\set{A}$ is denoted by $|\set{A}|$. The positive part of a real function is defined by $\left[f(x)\right]^+:=\max\left\{0, f(x)\right\}$. The notation that will be used in this paper is summarized in Table \ref{tab:notation}.
\begin{table}[t]
\centering
\centering
\caption{NOTATION SUMMARY}
\begin{tabular}{|p{0.17\linewidth}|p{0.75\linewidth}|}
% \begin{tabular}{|c|c|}
\hline
$\Ns$ & set of BSs with $|\Ns| = N$\\
$\Ks$ & set of UEs with $|\Ks| = K$\\
$\Ss$ & set of services with $|\Ss| = S$\\
$\Ws$ & set of MRUs with $|\Ws| = W$ \\ 
% $\Ks_n$ & set of \acp{UE} in \ac{BS} $n$ \\
$\Ss^{\ul}$ ($\Ss^{\dl})$ & set of \ac{UL} (\ac{DL}) services\\
% $\Ss_n^{\ul} (\Ss_n^{\dl})$ & set of  \acsp{UL} (\acsp{DL}) services  in \ac{BS} $n$\\
% $\Ws_n^{\ul} (\Ws_n^{\dl})$ & set of MRUs allocated to \acsp{UL} (\acsp{DL})  in \ac{BS} $n$\\
% $\Ws_{n,s}$ & set of MRUs allocated to service $s$ served by \ac{BS} $n$ \\
$\Ss_n$ & set of services served by the $n$th BS\\
$n_s$ & index of BS serving the $s$th service\\
$\mA$ & UE-to-service association matrix\\
$\mB$ & BS-to-service association matrix\\
$\mB^{\ul}$ ($\mB^{\dl}$) & BS-to-\ac{UL} (BS-to-\ac{DL}) association matrix\\
$\delta_t$($\delta_f$) & length of time duration (range of frequency) of an MRU\\
$W_t$($W_f$) & number of smallest time (frequency) units in MRU set $\Ws$\\
$\vw$ & fraction of resource allocated to services\\
$\vnu$ & cell load\\
$\vnu^{\ul}(\vnu^{\dl})$ & cell load in \ac{UL} (\ac{DL})\\
$\vp$ & transmit power allocated to services\\
$\vd$ & traffic demand of services\\
$\mH$ & channel gain matrix\\
$\mV$ & link gain coupling matrix \\
$\rho_s$ & per service QoS satisfaction level\\
$\rho$ & worst-case QoS satisfaction level\\
\hline
\end{tabular}
\label{tab:notation}
\vspace{-2em}
\end{table}
% The $k\times k$ identity matrix is denoted by $\mI_k$ and the $n\times k$ zero matrix is denoted by $\ma{0}_{n\times k}$. The $k$-dimensional all-ones (all-zeros) vector is denoted by $\ve{1}_k$ ($\ve{0}_k$), where the subscript $k$ is omitted if the dimension is clear from context. 
%For any $\vx\in\R^k$ and $c\in\R$, the notation $c\vx$ is used to denote $(c x_1, \ldots, c x_k)^T$, and $\vx + c$ is used to denote $\vx + c\ve{1}$. 

We consider an \ac{OFDM}-based wireless network system, consisting of a set of \acp{BS} $\Ns:=\{n: n=1, 2, \ldots, N\}$ and a set of \acp{UE} $\Ks:=\{k: k= 1,2, \ldots, K\}$. We assume that the network enables flexible duplex, where the resource in both frequency and time domains can be dynamically assigned to \ac{UL} and \ac{DL}. We define {\it \ac{MRU}} as the smallest time-frequency unit,  that has a length of $\delta_t$ seconds in time domain and a range of $\delta_f$ Hz in frequency domain. We consider a set of \acp{MRU}, denoted by $\Ws$,  consisting of $W_t$ smallest time units and $W_f$ smallest frequency units, and we have $W:=|\Ws| = W_t\cdot W_f$.

We assume that $K$ \acp{UE} generate a set of \ac{UL} and \ac{DL} services $\Ss:=\Ss^{\ul}\cup\Ss^{\dl}$ within the time duration of $W$ \acp{MRU} (i.e., $W_t\delta_t$ seconds).
% , where $\Ss^{\ul}$ with $\left|\Ss^{\ul}\right| = S^{\ul}$ and $\Ss^{\dl}$ with $\left|\Ss^{\dl}\right| = S^{\dl}$ denote the set of \ac{UL} services and the set of \ac{DL} services respectively, and we have $\Ss^{\ul}\cap \Ss^{\dl} = \emptyset$ and $S^{\ul} + S^{\dl} = S$. 
% Note that we allow \acp{UE} to generate multiple and different number of \ac{UL} and \ac{DL} services within the time duration. 
Let the \ac{UE}-to-service association matrix be denoted by $\mA\in\{0,1\}^{K\times S}$, where $a_{k,s} = 1$ means that the $s$th service is generated by the $k$th \ac{UE}, and $0$ otherwise. 
Let $\mB\in\{0,1\}^{N\times S}$ denote the \ac{BS}-to-service association matrix. To differentiate \ac{UL} and \ac{DL} services, we further define \ac{BS}-to-\ac{UL} and \ac{BS}-to-\ac{DL} association matrices, denoted by $\mB^{\ul}\in\{0,1\}^{N\times S}$ and $\mB^{\dl}\in\{0,1\}^{N\times S}$, respectively.
% , where $b^{\ul}_{n,s} = 1$ means that the $s$th service is associated to the $n$th BS in \ac{UL}, while $b^{\dl}_{n,s} = 1$ means that the $s$th service is associated to the $n$th BS in \ac{DL}. 
% Let the sets of services associated with \ac{BS} $n$ in \ac{UL} and \ac{DL} be denoted by $\Ss_n^{\ul}$ and $\Ss_n^{\dl}$ respectively. 
Let the set of services served by \ac{BS} $n$ be denoted by $\Ss_n$ and let the \ac{BS} associated with service $s$ be denoted by $n_s$.

Let $\vw:=[w_1, \ldots, w_S]^T\in [0,1]^S$ be a vector collecting the fraction of resource allocated to all services $s\in\Ss$. The {\it cell load}, defined as the fraction of occupied resource within a cell, is denoted by $\vnu = \mB\vw\in[0,1]^{N}$. The cell load in \ac{UL} and \ac{DL} are denoted by $\vnu^{\ul} = \mB^{\ul}\vw$ and $\vnu^{\dl} = \mB^{\dl}\vw$ respectively, and we have $\vnu = \vnu^{\ul} + \vnu^{\dl}$. 
We collect the transmit power (in Watt) allocated to all services in a vector $\vp := [p_1, \ldots, p_S]^T$. 
\subsection{Link Gain Coupling Matrix}\label{subsec:LinkGain}
We assume that average channel gains over $W$ \acp{MRU} from each \ac{TX} to each \ac{RX} are known, collected in $\mH := (h_{i,j})\in\RP^{(N+K)\times (N+K)}$. Note that the \acp{TX} and \acp{RX} include both \acp{UE} and \acp{BS}. 
Let $v_{l,s}$ denote the channel gain of the link between the \ac{TX} of link $l$ and the \ac{RX} of link $s$. If $l=s$, $v_{l,s}$ is the channel gain of link $s$, otherwise if $l\neq s$, $v_{l,s}$ is the channel gain of the interference link caused by service $l$ to $s$.
We define {\it link gain coupling matrix} $\mVt$ as  
\begin{equation}
\mVt:=(\vt_{l,s})\in\RN^{S\times S}, \mbox{ with } \vt_{l,s}:=v_{l,s}/v_{s,s},
\label{eqn:LinkGainV}
\vspace{-0.5ex}
\end{equation}
where $\vt_{l,s}$ is the ratio between the interference link gain from service $l$ to service $s$ and the serving link gain of $s$. 

An example is shown in Fig. \ref{fig:Coupling}, where we consider a system enabling downlink and uplink decoupling in \ac{5G} \cite{elshaer2014downlink}. The interference caused by \ac{UL} service $3$ (link $l_3$) to \ac{DL} service $1$ (link $l_1$) has a link gain of $v_{3,1} = h_{3, 4}$, i.e., the link gain between \ac{TX} $3$ (transmitter of $l_3$) and \ac{RX} $4$ (receiver of $l_1$). Given that the channel gain of $l_1$ is $h_{2,4}$, the interference coupling ratio is given by $\vt_{3,1} = h_{3,4}/h_{2,4}$. 
% 
% On the other hand, the interference links caused by \ac{DL} services $l_1$ and $l_2$ to \ac{UL} $l_3$ are the same since $l_1$ and $l_2$ share the same \ac{TX} $2$, thus we have $v_{1,3} = v_{2,3} = h_{2,1}$, and $\vt_{1,3} = \vt_{2,3} = h_{2,1}/h_{3,1}$, where $h_{3,1}$ is the channel gain of $l_3$.
	
\begin{Rem}[Incorporating different interference conditions]
Without loss of generality, we can modify $\mVt$ to take into account different interference conditions. For example, to allow self-interference cancellation we can define $\vt_{s,s} := 0$ for every $s\in\Ss$,  while to allow zero intra-cell interference we have $\vt_{l,s} := 0$ if $l$ and $s$ are associated with the same BS.
\label{rem:InterCond}
\vspace{-2ex}
\end{Rem}

\begin{figure}[t]
\centering
\includegraphics[width=.63\columnwidth]{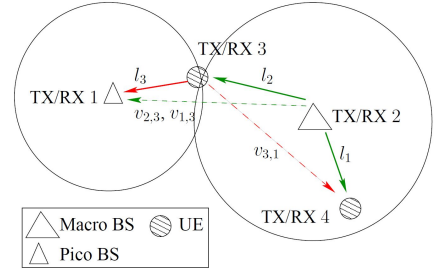}
%\scalebox{.55}{\input{InterCellInterferenceULDL_PerLink.pstex_t}}
\vspace{-2ex}
\caption{Example: Interference link gain.}
\label{fig:Coupling}
\vspace{-3ex}
\end{figure}

\subsection{Quality of Service Metric}\label{subsec:QoS}
In \cite{liao2016dynamic} we assume that the probability that $l$ causes \ac{ICI} to $s$ associated with a different \ac{BS} is approximated by the fraction of its allocated resource $w_l$, which leads to
\begin{equation}
\Pr\left\{l \mbox{ interferes } s| n_l\neq n_s\right\}\approx w_l \mbox{ for } l,s \in\Ss.
\label{eqn:InterApprox_1}
\end{equation}
The average \ac{SINR}\footnote{Note that  $\vt_{l,s}$ is computed with average channel gain over $W$ \acp{MRU}. Thus, \eqref{eqn:SINR_old} is the ratio between average received signal strength and average received interference, rather than the actual average \ac{SINR}. Since we do not assume to know the distribution of the channel gain, here we use \eqref{eqn:SINR_old} to approximate the average \ac{SINR}.}  of $s\in\Ss$ is approximated by
%\begin{align}
%\sinr_s & \approx \dfrac{p_s}{\sum_{l\in\Ss} \vt_{l,s} p_s w_s + \sigma_s^2} \nonumber\\
%& = \dfrac{p_s}{\left[\mVt\diag(\vp)\vw + \ve{\sigma}\right]_s}
%\end{align}
%
\begin{equation}
\sinr_s  \approx \dfrac{p_s}{\sum\limits_{l\in\Ss} \vt_{l,s} p_l w_l + \dfrac{\sigma_s^2}{v_{s,s}}}= \dfrac{p_s}{\left[\mVt^T\diag(\vw)\vp + \tilde{\ve{\sigma}}\right]_s},
\label{eqn:SINR_old}
\end{equation}
where $\tilde{\ve{\sigma}} := \left[\sigma_1^2/v_{1,1}, \sigma_2^2/v_{2,2}, \ldots, \sigma_S^2/v_{S,S}\right]^T$,  $\sigma_s^2$ denotes the noise power in the receiver of $s$. Note that in \eqref{eqn:SINR_old} $w_l$ serves as a probability. The interference condition is taken into account in $\vt_{l,s}$ as illustrated in Remark \ref{rem:InterCond}. % For example, we have $\vt_{l,s} := 0$ for $n_l=n_s$ if zero intra-cell interference is assumed (due to the orthogonal resource allocation among users served by the same \ac{BS}). 
However, the approximations \eqref{eqn:InterApprox_1} and \eqref{eqn:SINR_old} are only valid under the assumption that each \ac{MRU} is considered to be \lq\lq equal\rq\rq \ for all the services to be allocated, namely, the position of resource is not specified for \ac{UL} or \ac{DL}. Unfortunately, such assumption results in a high probability of \ac{IMI}. In the following we introduce an improved \ac{SINR} model based on a simple \ac{UL}/\ac{DL} resource positioning strategy to reduce \ac{IMI}. 

Recall that conventional \ac{TDD} or \ac{FDD} specifies a set of resource for \ac{UL} and \ac{DL} respectively to prevent \ac{IMI}. With flexible duplex, the challenge is to allow different resource partitioning between \ac{UL} and \ac{DL} in each cell, while limiting the probability of generating \ac{IMI}. Let us take an example, cell $m$ with \ac{UL} load $\nu_m^{\ul}$ and cell $n$ with \ac{DL} load $\nu_n^{\dl}$ share same set of available resource. It is obvious that the minimum overlapping area between \ac{UL} resource in cell $m$ and \ac{DL} resource in cell $n$ is $\left[\nu_m^{\ul} + \nu_n^{\dl} -1\right]^+$, which can be easily achieved by allocating the set of resource to \ac{UL} traffic in cell $m$ in some priority order while allocating the same set of resource to \ac{DL} traffic in cell $n$ in reverse order. 
% A detail example is given in Appendix \ref{sec:app_SelectMRU}.

% In this paper, we consider a simple strategy which defines a priority order for allocating the \ac{UL} and \ac{DL} traffic respectively. The trick is to define the priority of each \ac{MRU} for allocating \ac{UL} and \ac{DL} traffic {\it in reverse order},  such that the reused (overlapping) resource area generating inter-cell inter-mode interference is minimized. A detail example is given in Appendix \ref{sec:app_SelectMRU}.
%, while defining the priority for allocating to \ac{DL} traffic to the same set of resources {\it in reverse order}. 
% Given the number of resources allocated to \ac{UL} and \ac{DL} in each cell, the position of resources allocated to \ac{UL} and \ac{DL} is selected based on the priority orders for \ac{UL} and \ac{DL} respectively. 
%

Given the aforementioned strategy, to derive the interference coupling matrix that incorporates the probability that a link causes \ac{ICI} to another, we introduce a {\it reuse factor coupling matrix} $\mC(\vw)$ depending on $\vw$. Let $x_s\in\{\text{u},\text{d}\}$ denote the \ac{UL} or \ac{DL} traffic type of service $s\in\Ss$, and recall that $n_s$ denotes the serving \ac{BS} of $s$, $\mC(\vw)$ is defined as  
\begin{align}
\vspace{-0.3ex}
\mC(\vw) & :=\mC:= \left(c_{l,s}\right) \in\RN^{S\times S}, \label{eqn:MatrixC_0} \\
c_{l,s} & := 
\begin{cases}
 \left[ \left(\nu_{n_l}^{(x_l)} + \nu_{n_s}^{(x_s)}-1\right)/\nu_{n_s}^{(x_s)}\right]^+ & \mbox{  if } x_l \neq x_s \\
 \min\left\{1, \nu_{n_l}^{(x_l)}/\nu_{n_s}^{(x_s)}\right\} & \mbox{ if } x_l = x_s,\nonumber
\end{cases}
\label{eqn:MatrixC}
\vspace{-0.3ex}
\end{align} 
where the load of cell $n_s$ occupied by traffic type $x_s$ is computed by $\nu_{n_s}^{(x_s)}:=\left[\mB^{(x_s)}\vw\right]_{n_s}$. In general, $c_{l,s}$ is defined as the ratio of the overlapping area on the resource plane between the load of cell $n_l$ serving traffic type $x_l$ and the load of cell $n_s$ serving traffic type $x_s$ to the load of cell $n_s$ serving traffic type $x_l$.

With $\mC(\vw)$ in hand, given the power vector $\vp$, we can modify \eqref{eqn:SINR_old} and derive the \ac{SINR} of service $s\in\Ss$ as  
\begin{equation}
\vspace{-0.3ex}
\sinr_s(\vw)  \approx \dfrac{p_s}{\left[\left(\mC(\vw)\circ\mVt\right)^T\diag(\vp)\vw + \tilde{\ve{\sigma}}\right]_s},
\label{eqn:SINR_new}
\vspace{-0.3ex}
\end{equation}
where with a slight abuse of notation, $\mX \circ \mY$ denotes the Hadamard (entrywise) product of matrices $\mX$ and $\mY$. Note that the first term in the denominator is the interference power received by service $s$ divided by the channel gain of $s$, and it is equivalent to $\sum_l c_{l,s}w_l v_{l,s}p_l/v_{s,s}$, where $c_{l,s} \cdot w_l$ approximates the probability that service $l$ causes interference to service $s$.
% (details are given in Appendix \ref{sec:app_reuseMatrix}).

The maximum achievable number of bits for service $s\in\Ss$ within the time span of resource set $\Ws$ is  
\begin{equation}
\vspace{-0.3ex}
\eta_s(\vw) = \delta_t\delta_f W w_s \log \left(1 + \sinr_s(\vw)\right),
\label{eqn:throughput}
\vspace{-0.3ex}
\end{equation}
where the unit of $\delta_t\delta_f$ is Hz$\cdot$s/\ac{MRU}, while $W w_s$ is the number of \acp{MRU} allocated to $s$. 

Assuming that the nonzero traffic demands $\vd:=(d_1, \ldots, d_S)^T\in\RP^S$ is known, where $d_s$ is defined as number of required bits of $s$ during the time span of $\Ws$, we introduce {\it per service \ac{QoS} satisfaction level},  written as 
\begin{equation}
\rho_s(\vw) = \eta_s(\vw)/d_s, \ s\in\Ss.
\label{eqn:QoS}
\end{equation}

\section{Problem Formulation}\label{sec:Prob}
The objective is to partition the resource set $\Ws$ in each cell $n\in\Ns$ into three subsets: resource for  \ac{UL}, resource for \ac{DL}, and blanked resource\footnote{Under certain conditions,  enhanced interference mitigation can be achieved by muting partial resources in some cells. However, it is also possible that the optimal solution returns an empty set of the blanked resource.}, respectively, to maximize {\it the worst-case \ac{QoS} satisfaction level}, defined as
\begin{equation}
\rho(\vw):= \min_{s\in\Ss}\rho_s(\vw).
\label{eqn:rho}
\end{equation}
All demands of the services are feasible, when $\rho(\vw)\geq 1$.

% The problem to find out what is the maximum worst-case \ac{QoS} satisfaction level $\rho^{\ast}$ (implied by \eqref{eqn:problem_1a} and \eqref{eqn:problem_1b} in the following Problem \ref{prob:maxminQoS}), stated as follows.

We formulate the problem in Problem \ref{prob:maxminQoS}, where \eqref{eqn:problem_1a} and \eqref{eqn:problem_1b} imply the objective of maximizing the worst-case \ac{QoS} satisfaction level $\rho^{\ast}$, and \eqref{eqn:loadconstraints} is the per-cell load constraint.

%\begin{problem}
%\begin{subequations}
%\label{eqn:prob_1}
%\begin{align}
%\maximize_{\vw\in\RN^S} \ & \rho\\
%\mbox{subject to } &  \rho(\vw):= \min_{s\in\Ss}\lambda_s(\vw) \label{eqn:rho},\\
%& \eqref{eqn:MatrixC}, \eqref{eqn:SINR_new}, \eqref{eqn:throughput}, \eqref{eqn:QoS},\\
%& g(\vw):=\|\mB\vw\|_{\infty} \leq 1,\label{eqn:loadconstraints}
%\end{align}
%\end{subequations}
%where \eqref{eqn:loadconstraints} implies the load constraint per cell. 
%\label{prob:maxminQoS}
%\end{problem}		

% Problem \ref{prob:maxminQoS} can be rewritten in the general form as below.
% remains equivalent when replacing \eqref{eqn:rho} by

\begin{problem} 
\begin{subequations}
\label{eqn:prob_1}
\begin{align}
\maximize_{\vw\in\RN^S, \rho\in\RN} \ & \rho \label{eqn:problem_1a}\\
\mbox{s.t. } & \vw  \geq \rho \vf(\vw), \label{eqn:problem_1b}\\
 & g(\vw) :=\|\mB\vw\|_{\infty} \leq 1,\label{eqn:loadconstraints}
\end{align}
\end{subequations}
where the vector-valued function $\vf$ is defined by
\begin{subequations}
\label{eqn:define_vecf}
\begin{align}
 \vf: \RN^{S}\to \RP^{S}:  & \vw\mapsto\left[f_1(\vw), \ldots, f_S(\vw)\right]^T, \label{eqn:define_f}\\
 \mbox{where }   f_s(\vw) := & \frac{d_s}{\delta_t\delta_f W \log\left(1 + \sinr_s(\vw)\right)}. \label{eqn:define_fs}
\end{align}
\end{subequations}
\label{prob:maxminQoS}
\end{problem}

In \cite{liao2016dynamic}, we show that with conventional model of \ac{SINR} \eqref{eqn:SINR_old}, Problem \ref{prob:maxminQoS} is equivalent to solve a nonlinear system of equations such that $\vw = \rho \vf(\vw)$, $g(\vw) = 1$ and that $\rho$ is maximized. It is worth mentioning that, with the modified models of interference coupling \eqref{eqn:MatrixC_0} and \ac{SINR} \eqref{eqn:SINR_new}, Problem \ref{prob:maxminQoS} is a multi-variate nonconvex optimization problem. Moreover, the constraint \eqref{eqn:problem_1b} is neither convex nor continuously differentiable, and Problem \ref{prob:maxminQoS} is not necessarily equivalent to the nonlinear system of equations.
% Therefore, many existing algorithms for nonconvex optimization problems with nonconvex smooth objective functions or constraints are inapplicable \cite{scutari2014parallel}. 

%
In Section \ref{sec:Algor_SAFP} we provide algorithmic solution to Problem \ref{prob:maxminQoS}, denoted by $\vw^{\ast}$. The per-cell fraction of resource to allocated to \ac{UL} and \ac{DL} are then obtained as $\vnu^{\ul,\ast}=\mB^{\ul}\vw^{\ast}$ and $\vnu^{\dl, \ast}=\mB^{\dl}\vw^{\ast}$, respectively. If $\rho^{\ast}:=\rho(\vw^{\ast})\geq 1$, all demands are feasible. % and $\rho^{\ast}$ is the maximal achievable worst-case QoS satisfaction level.
However, if $\rho^{\ast}<1$, the solution to Problem \ref{prob:maxminQoS} is not a good operating point, since the demands of all services are infeasible. In other words, all users are unsatisfied. Therefore, a further question arises: {\it how can we transform the desired demands in Problem \ref{prob:maxminQoS} from infeasible to feasible?} 
One of the factors causing infeasible demand is the bottleneck services.  
In Section \ref{sec:Algor_RMDI} we modify Problem \ref{prob:maxminQoS} by dedicating partial resources for bottleneck services, while muting them for others, and develop an algorithm with heuristic strategies. 
%
% heuristic algorithms are developed in Section \ref{subsec:Heuristics}. 
%
\begin{Rem}[New challenge due to complex interference coupling]
Problem \ref{prob:maxminQoS} is formulated along similar lines to our previous work \cite[Problem 2a]{liao2016dynamic}. However, in \cite{liao2016dynamic}, the received interference in \ac{SINR} \eqref{eqn:SINR_old} is an affine function of $\vw$, which further leads to some nice properties of $\vf$ (as shown in Lemma \ref{lem:f_SIF}). In this paper, because we introduce more complex interference coupling \eqref{eqn:MatrixC_0} and the resulting modified \ac{SINR} model \eqref{eqn:SINR_new}, the desired properties of $\vf$ do not exist, which brings new challenge with developing efficient algorithmic solution.
\end{Rem}
%
%%%%%%%%%%%%%%%%%%%%%%%%%%%%%%%%%%%%%%%%%%%%%
\section{Successive Approximation of Fixed Point}\label{sec:Algor_SAFP}
In this section, we first provide background information about the mathematical tool to solve the problem. Then, we propose a novel efficient algorithm \ac{SAFP} to find a feasible point of $\vw$ with good, if not optimal, objective value of $\rho^{\ast}$.
% We present an efficient suboptimal solution for maximizing the worst-case QoS satisfaction level, which finds a feasible point of $\vw$ with good, if not optimal, objective value of $\rho^{\ast}$. If $\rho^{\ast}<1$,  we provide in Section \ref{subsec:Heuristics} two heuristic approaches to improve the number of services with satisfied demands:  {\it removal of bottleneck services} and {\it resource muting for dominant interferer}.
%
% \subsection{Suboptimal Solution to Problem \ref{prob:maxminQoS}}\label{subsec:SolutionProb_1}
\subsection{Background Information and Previous Results}\label{subsec:background}
With the conventional \ac{SINR} model in \eqref{eqn:SINR_old}, $\vf$ defined in \eqref{eqn:define_vecf} has the following property.
  
% If we fix the reuse factor coupling matrix $\mC'\in\RN^{S\times S}$, i.e., if $\mC'$ consists of fixed constant coefficients, the received interference in \eqref{eqn:SINR_new} is still an affine function of $\vw$ as in the conventional \ac{SINR} model in \eqref{eqn:SINR_old}. Furthermore, $\vf_{\mC'}:\RN^{S}\to\RP^{S}$ defined in \eqref{eqn:define_f} and \eqref{eqn:define_fs} (but with fixed $\mC'$ instead of $\mC(\vw)$) has the following property.
\begin{lemma}[{\cite[Lemma 1]{liao2016dynamic}}]
With \ac{SINR} defined in \eqref{eqn:SINR_old}, $\vf:\RN^{S}\to\RP^{S}$ is a \ac{SIF} (see Appendix \ref{sec:app_Background} for definition).
\label{lem:f_SIF}
\end{lemma}
%\begin{lemma}
%If $\mC(\vw)$ in \eqref{eqn:SINR_new} is replaced by $\mC'$ consisting of constant coefficients, $\vf_{\mC'}:\RN^{S}\to\RP^{S}$ is \ac{SIF} (see Appendix \ref{sec:app_Background} for definition).
%\label{lem:f_SIF}
%\end{lemma}
%\begin{proof}
%The proof is omitted here since it is along the same lines as \cite[Lemma 1]{liao2016dynamic}. The essential step is to show that the nonlinear mapping $f_{\mC',s}:\RN^{S}\to\RP$ for every $s\in\Ss$ is concave, and that positive concave function is shown to be \ac{SIF} in \cite[Proposition 1]{cavalcante2016elementary}.
%\end{proof}
%
Knowing that $\vf$ is \ac{SIF}, and that $g:\RP^S\to\RP$ in \eqref{eqn:loadconstraints} is a monotonic norm,  we encounter the same type of problem as \cite[Problem 2a]{liao2016dynamic}. The following proposition is provided based on the previous result \cite[Theorem 1]{liao2016dynamic}, which gives rise to an algorithmic solution to Problem \ref{prob:maxminQoS} with conventional \ac{SINR} model based on the fixed point iteration scheme.
\begin{proposition}
Suppose \ac{SINR} is modeled with \eqref{eqn:SINR_old}, and 
\begin{itemize}
\item $\vf:\RN^S\to\RP^S$ is \ac{SIF}, 
\item $g:\RP^S\to \RP$ is monotonic, and homogeneous with degree $1$ (i.e., $g(\alpha \vx) = \alpha g(\vx)$ for all $\alpha >0$)
\end{itemize}
There exists a unique solution to Problem \ref{prob:maxminQoS}, denoted by $\{\vw^{\ast}, \rho^{\ast}\}$, where $\vw^{\ast}$ can be obtained by performing the following fixed point iteration:
\begin{equation}
\vw^{(t+1)} = \frac{\vf\left(\vw^{(t)}\right)}{g\circ\vf\left(\vw^{(t)}\right)}, t\in\NN,
\label{eqn:fixedpoint_w}
\end{equation}
where with a slight abuse of notation, $g\circ\vf$ denotes the composition of functions $g$ and $\vf$. The iteration in \eqref{eqn:fixedpoint_w} converges to $\vw^{\ast}$, and we have $\rho^{\ast} = 1/g\circ\vf(\vw^{\ast})$ and $g(\vw^{\ast})=1$.
\label{prop:solutionFixedC}
\end{proposition}
\begin{proof}
The proof is omitted here since it uses our previous result \cite[Theorem 1]{liao2016dynamic} and is along the same lines as \cite[Proposition 1]{liao2016dynamic}.  
\end{proof}

% ----------------------------------------------------- 
\subsection{Successive Approximation of Fixed Point}\label{subsec:AlterOpt}
Proposition \ref{prop:solutionFixedC} provides an algorithmic solution to Problem \ref{prob:maxminQoS} with \ac{SINR} \eqref{eqn:SINR_old}, by utilizing the properties of \ac{SIF}.  Unfortunately, with the modified \ac{SINR} in \eqref{eqn:SINR_new}, $\vf$ is not \ac{SIF} because the coupling matrix $\mC(\vw)$ depends on $\vw$ in a non-monotonic and non-differentiable manner.  However, it is easy to show that by replacing $\mC(\vw)$ in \eqref{eqn:SINR_new} with some approximation $\mC':=\mC(\vw')$ computed with fixed $\vw'$,  the \ac{SINR} in \eqref{eqn:SINR_new} falls into the same class as \eqref{eqn:SINR_old}, and the approximated problem can be solved by Proposition \ref{prop:solutionFixedC} with $\vf(\vw)$ replaced by $\vf_{\mC'}(\vw):=\vf(\vw, \mC(\vw'))$.

Therefore, our essential, natural idea is to efficiently compute a suboptimal solution of Problem \ref{prob:maxminQoS} by solving a sequence of (simpler) max-min fairness subproblems whereby the noncontractive mapping $\vf$ is replaced by suitable contraction approximation $\vf_{\mC'}$. These subproblems can be solved with Proposition \ref{prop:solutionFixedC}.

More specifically, the proposed \ac{SAFP} algorithm consists in solving a sequence of approximations of Problem \ref{prob:maxminQoS} in the form 
\begin{equation}
\maximize\limits_{\vw\in\RN^S, \rho\in\RN}  \rho;
\mbox{ s.t. } \vw\geq \rho \vf_{\mC'}\left(\vw\right); \ g(\vw)\leq 1,
\label{eqn:problem_fixedC}
\vspace{-0.7ex}
\end{equation}
where $\vf_{\mC'}(\vw)$ represents approximation of $\vf(\vw)$ at the current iterate $\vw'$. The unique solution to \eqref{eqn:problem_fixedC} can be obtained by the fixed point iteration \eqref{eqn:fixedpoint_w}, with $\mC(\vw)$ replaced by $\mC(\vw')$. 

Unfortunately, due to the complexity of $\mC(\vw)$, the convergence of \ac{SAFP} to a limit point cannot be guaranteed, since multiple fixed points can exist in the system where the inequality sign in \eqref{eqn:problem_1b} is replaced by the equality sign. Different initial values of $\hat{\vw}$ may lead to different fixed points. {\it Moreover, the solution to the system of nonlinear equations may not be the optimal solution to the original problem of maximizing the minimum, due to the nonmonotonicity of the mapping $\vf$ when including $\mC$ into the interference model.} 
Thus, we design the searching algorithm to guarantee the utility increase with initial values of $\{\rho^{\ast}, \vw^{\ast}\}$, maximum number of random initiation $N_{\text{max}}$, and algorithm stopping criterion depending on the maximum number of iterations $N_{\iter}$  and the distance threshold $\epsilon$, illustrated as below.
\begin{itemize}
\item The algorithm runs for $N_{\max}$ times, each with a different random initialization of $\hat{\vw}$ and the corresponding $\mC(\hat{\vw})$. 
\item For each initialization $\hat{\vw}_n$, $n= 1, 2, \ldots, N_{\max}$, we iteratively perform the fixed point iteration in \eqref{eqn:fixedpoint_w} with $\vf(\vw)$ replaced by $\vf_{\hat{\mC}_n}(\vw)$ where $\hat{\mC}_n := \mC(\hat{\vw}_n)$. The iteration stops if the number of iterations exceeds $N_{\iter}$ or the distance yields $\|\vw'-\vw\|\leq\epsilon$ and returns the solution $\{\vw', \rho'\}$ with respect to the $n$th random initialization. The solution is updated with $\vw^{\ast}\leftarrow\vw'$, $\rho^{\ast}\leftarrow\rho'$ if $\rho'>\rho^{\ast}$. 
%The converged solutions $\{\vw_n, \rho_n\}$ are stored for the $n$th random initialization as candidate solutions. 
%\item The resource allocation $\vw^{\ast}$ is selected with respect to the maximal $\rho^{\ast}$ among the candidate set of solutions $\{(\vw_n, \rho_n): n=1, \ldots, N_{\text{max}}\}$. 
\end{itemize}
The proposed \ac{SAFP} algorithm is summarized in Algorithm \ref{algo:AlgorithmUpdateW}.

Although the convergence of \ac{SAFP} to a global optimum cannot be guaranteed and heuristics are introduced, numerical results in Section \ref{sec:Numeric} (e.g., Fig. \ref{fig:SAFP_initialization}) show that each random initialization converges to a fixed point, and with limited number of initializations, the algorithm finds a suboptimal, if not optimal, solution among multiple fixed points.

% Since Problem \ref{prob:maxminQoS} is a multi-variate nonlinear nonconvex optimization problem, it is difficult to derive the optimal solution with a computationally efficient algorithm to meet the needs of practical implementation. Thus, 

%We propose to introduce the heuristic based on alternating optimization, by utilizing the results in Proposition \ref{prop:solutionFixedC}. The algorithm summarized in Algorithm \ref{algo:AlgorithmUpdateW} consists of $N_{\iter}$ iterations of two alternating steps:
%\begin{itemize}
%\item[{\it S1.}]  Optimize $\vw$ using updated $\mC'$ with \eqref{eqn:fixedpoint_w}; 
%\item[{\it S2.}]  Computing $\mC$ using updated $\vw'$ with \eqref{eqn:MatrixC}. 
%\end{itemize}
%Note that Problem \ref{prob:maxminQoS} is nonconvex, the above alternating optimization may not converge to a point. Thus, the algorithm stopping criterion depends on the maximum number of iterations $N_{\iter}>1$ and the 

\setlength{\textfloatsep}{0pt}
\begin{algorithm}[t]
\SetKwData{band}{$\vw'$}
\SetKwData{coupling}{$\mC'$}
% \SetKwData{Up}{up}
%\SetKwFunction{UpdateBandwidth}{UpdateBandwidth}
%\SetKwFunction{UpdateCouplingMatrix}{UpdateCouplingMatrix}
\SetKwInOut{Input}{input}
\SetKwInOut{Output}{output}
\caption{\ac{SAFP} algorithm for resource partitioning}
%\Input{$\vw'\leftarrow \hat{\vw}\in\RP^{S}$, $\vw\leftarrow\ve{0}$, $\mC' \leftarrow \mC(\vw')\in\RP^{S\times S}$, $i\leftarrow 0$, $N_{\iter}>1$, $\Cs_{\vw}^{(i)}\leftarrow\emptyset$, $\Cs_{\rho}^{(i)}\leftarrow\emptyset$,  $\epsilon$}
\Input{$i\leftarrow 1$, $N_{\text{max}}>1$, $N_{\iter}>1$, $\epsilon>0$, $\rho^{\ast}\leftarrow 0$, $\vw^{\ast}\leftarrow \ve{0}$
% $\Cs_{\vw}^{(0)}\leftarrow\emptyset$,  $\Cs_{\rho}^{(0)}\leftarrow\emptyset$
}
\Output{$\{\vw^{\ast}, \rho^{\ast}\}$}
\BlankLine
\While{$i\leq N_{\text{max}}$}{
random initialization of $\vw'$; $\mC'\leftarrow\mC(\vw')$\;
$j \leftarrow 0$, $\vw\leftarrow\ve{0}$\;
$\Delta^{(j)} \leftarrow \|\vw'-\vw\|_{\infty}$; \ $\vw^{(j)}\leftarrow \vw'$\; 
 \While{$j\leq N_{\text{iter}}$ or $\Delta^{(j)}\geq\epsilon$}{
 \emph{\% solving approximated subproblem with $\mC'$}\;
     \While{$\|\vw'-\vw\|_{\infty}\geq\epsilon$}{
		 $\vw\leftarrow\vw'$\;
      \band $\leftarrow \vf_{\mC'}(\vw)/g\circ\vf_{\mC'}(\vw)$ \;
     }
		\emph{\% Update $\mC$ with optimized $\vw'$}\;
		$\vw^{(j+1)} \leftarrow \vw'$\;
			$\mC^{(j+1)} = $ \coupling $\leftarrow \mC(\vw')$ \;
			%
			%\emph{\% Store the candidate solution}\;
		  %$\Cs_{\vw}^{(i+1)} \leftarrow \Cs_{\vw}^{(i)}\cup \{\vw^{(i+1)}\}$\; 
			%$\Cs_{\rho}^{(i+1)} \leftarrow \Cs_{\rho}^{(i)}\cup \{\rho^{(i+1)}\}$\;
			$\Delta^{(j+1)} \leftarrow \|\vw^{(j+1)} -\vw^{(j)}\|_{\infty}$\;
			$j \leftarrow j+1$\;
			}
		% \emph{\% Compute actual utility $\rho'$ based on \eqref{eqn:rho}}\;
	  $\rho' = \rho'(\vw')\leftarrow \min_{s\in\Ss} w_s' / f_{\mC',s}(\vw')$\;
		%\emph{\% save candidate solution for each initialization}\;
		\emph{\% update the solution if $\rho'$ exceeds the stored value}\;
		\If{$\rho'>\rho^{\ast}$}{
		$\rho^{\ast}\leftarrow \rho'$\;
		$\vw^{\ast}\leftarrow \vw'$\;
		}
		%$\Cs_{\vw}^{(i)} \leftarrow \Cs_{\vw}^{(i-1)}\cup \{\vw'\}$\;
		%$\Cs_{\rho}^{(i)} \leftarrow \Cs_{\rho}^{(i-1)}\cup \{\rho'\}$\;
		$i \leftarrow i+1$\;	
}
% choose the best candidate
%\emph{\% Choose the best candidate solution}\;
%$\rho^{\ast}\leftarrow \max_{\rho\in\Cs_{\rho}^{(i)}}\rho$\; 
%$\vw^{\star}\leftarrow \argmax_{\vw\in\Cs_{\vw}^{(i)}} \rho(\vw)$.
\label{algo:AlgorithmUpdateW}
\end{algorithm}

% ----------------------------------------------

\section{Resource Muting for Dominant Interferer}\label{sec:Algor_RMDI}
The proposed \ac{SAFP} finds a feasible point of $\vw^{\ast}$ with suboptimal, if not optimal, objective value of $\rho^{\ast}$. If $\rho^{\ast}\geq 1$,  the obtained $\vw^{\ast}$ provides fairness on the services, and the demands of all services are feasible. However, if $\rho^{\ast}<1$,  $\vw^{\ast}$ is not a good operating point  since the traffic demands of all services are infeasible. Therefore, in this section we focus the following question: {\it how can we transform the desired demands in Problem \ref{prob:maxminQoS} from infeasible to feasible?}
In \cite{takahashi2006efficient}, the authors propose a removal selection criterion for an infeasible \ac{DL} power control problem, that removes sequentially the bottleneck services until the demands for all the remaining services are feasible. However, is there a method of further increasing $\rho^{\ast}$ without removal of services? Motivated by coordinated muting using \ac{ABS} for time domain intercell interference coordination introduced in \cite{3gpp_ts36133}, we are interested in exploring the tradeoff between resource utilization and interference reduction by introducing the resource muting in flexible duplex.  
\vspace{-0.5ex}
\subsection{Modified Load Constraints Incorporating Resource Muting}
The key concept is to sequentially reserve some resource in a cell for the dominant interferer, while muting them in the cells strongly impacted by the interferer. To this end, we rank the services based on the interference level that they generate to others, given by 
\vspace{-0.7ex}
\begin{equation}
I_s(\vw)  :=  \left(\vc_{s}'{\ve{\vt}_{s}'}^T\right) p_s w_s, \mbox{ for } s\in\Ss,
\label{eqn:InterferenceMetric}
\vspace{-0.5ex}
\end{equation}
 %To this end, we define a metric to detect the dominant interferer in the system, as the sum of received interference from other cells  and the generated interference to other cells, given by
%\begin{align}
%I_s(\vw)  := &  \left[\left(\mC(\vw)\circ\mVt\right)^T\diag(\vp)\vw\right]_s + \nonumber \\
 %& \left(\vc_{s}'{\ve{\vt}_{s}'}^T\right) p_sw_s, \mbox{ for } s\in\Ss 
%\label{eqn:InterferenceMetric}
%\end{align}
where $\vc_{s}':=\row_s \mC(\vw)$ denotes the $s$th row of $\mC(\vw)$, and $\ve{\vt}_{s}':=\row_s \mVt$ denotes the $s$th row of $\mVt$. % The first and the second terms in \eqref{eqn:InterferenceMetric} reflect the received interference in $s$  and the generated interference by $s$,  respectively.  

Moreover, to prevent the waste of resource, we select the strongly affected cells to mute their resource. The set of cells to mute the resource reserved for $s$ is selected by 
\begin{equation}
\vspace{-0.5ex}
\Ms_s  := \{m\in\Ns\setminus\{n_s\}: J_{s,m}(\vw)\geq \alpha  \},
\label{eqn:InterferenceCellMetic_0}
\end{equation}
where $\alpha$ is a threshold and $J_{s,m}(\vw)$ is the interference generated from service $s$ to a cell $m\neq n_s$, defined as
\begin{equation}
 J_{s,m}(\vw)  :=   \left[\mB\left(\vc_{s}'\circ\ve{\vt}_{s}'\right)^T\right]_m p_s w_s.
\label{eqn:InterferenceCellMetic}
\end{equation}
%\begin{align}
%\Ms_s & := \{m\in\Ns\setminus\{n_s\}: J_{s,m}(\vw)\geq \alpha  \} \\ 
%\mbox{where } J_{s,m}(\vw) & :=   \left[\mB\left(\vc_{s}'\circ\ve{\vt}_{s}'\right)^T\right]_m p_s w_s
%\label{eqn:InterferenceCellMetic}
%\vspace{-0.5ex}
%\end{align}

%If $s$ is selected as the dominant interferer, to incorporate the resource muting scheme, we need to introduce additional load constraint to Problem \ref{prob:maxminQoS}, given by
%\begin{equation}
%g_s'(\vw) := \max_{m\in\Ms_s} w_s + \sum_{l\in\Ss_m} w_l \leq 1. 
%\label{eqn:load_const_2}
%\end{equation}
%\eqref{eqn:load_const_2} implies that the total resource allocated to each cell $m\in\Ms_s$ is less or equal to $1-w_s$, leaving resource $w_s$ muted for the dominant interferer $s$. 
If a set of dominant interferers $\overline{\Ss}$ is chosen, and for each $s\in\overline{\Ss}$ a subset of the cells $\Ms_s$ is selected to mute resource $w_s$, then, in each cell we have the load constraint
\begin{equation}
g'_m(\vw) :=\sum_{s\in\overline{\Ss}} \indication{m\in\Ms_s} w_s  + \sum_{l\in\Ss_m} w_l \leq 1, \mbox{ for } m\in\Ns,
\label{eqn:PerCellLoadConstraint}
\vspace{-0.5ex}
\end{equation}
where $\indication{\cdot}$ is the indication function, the first term is the total amount of resource to be muted in cell $m$, and the second term is the amount of available resource for services in $m$. 
 
Since $g'_m(\vw)\leq 1$ needs to be held for every $m\in\Ns$, the load constraint can be rewritten as
\begin{equation}
g'(\vw) := \max_{m\in\Ns} g_m'(\vw)\leq 1.
\label{eqn:jointLoadConstraint}
\vspace{-0.5ex}
\end{equation} 

Note that without the muting scheme, i.e., if $\bar{\Ss} = \emptyset$, the first term in \eqref{eqn:PerCellLoadConstraint} is zero and \eqref{eqn:jointLoadConstraint} is equivalent to the per-cell load constraints in \eqref{eqn:loadconstraints}.
  %
%We can write \eqref{eqn:load_const_3} and the per-cell load constraint \eqref{eqn:loadconstraints} in the unified term as defined below
%\begin{equation}
%g'_s(\vw) := \max\{g(\vw), \phi_s(\vw)\}\leq 1
%\label{eqn:jointLoadConstraint}
%\end{equation}
\subsection{Design of Heuristic Algorithm}
It is obvious that the modified $g'$ is also monotonic and homogeneous with degree $1$, which enables leverage of Proposition \ref{prop:solutionFixedC} to solve the modified Problem \ref{prob:maxminQoS}, with $g(\vw)$ replaced by $g'(\vw)$ to incorporate the resource reservation and muting strategy.  

Compared to the solution to the original Problem \ref{prob:maxminQoS},  resource muting may not necessarily improve the desired utility $\rho$, because muting of $w_s$ in cell $m\in\Ms_s$ may lead to waste of resource. Therefore, we develop a heuristic algorithm \ac{RMDI} to guarantee a utility that is no less than the $\rho$ derived in Algorithm \ref{algo:AlgorithmUpdateW}. The Algorithm is described briefly in the following steps.
\begin{itemize}
\item[1.] Derive $\vw^{(0)} = \vw^{\ast}$ to Problem \ref{prob:maxminQoS} with Algorithm \ref{algo:AlgorithmUpdateW} and compute the corresponding $\rho^{(0)} = \rho^{\ast}$. 
\item[2.] Compute $I_s(\vw^{\ast})$ and rank the services based on $I_s$. Let $q_s$ denote the rank of $s$, e.g., the maximum interferer $\hat{s}:=\argmax_s I_s$ has a rank of $q_{\hat{s}} = 1$.  Set $k= 1$. 
\item[3.] Add the service with highest rank into $\bar{\Ss}^{(k)}$, e. g., $\bar{\Ss}^{(k)} = \{s: q_s\leq k\}$. 
\item[4.] Solve modified Problem \ref{prob:maxminQoS} with $\bar{\Ss}^{(k)}$ using Algorithm \ref{algo:AlgorithmUpdateW} (with $g$ replaced by $g'$), derive $\vw^{(k)}$ and $\rho^{(k)}$.
\item[5.] If $\rho^{(k)}\geq \rho^{(k-1)}$, increment $k$ and go back to Step 3; otherwise stop the algorithm.
\item[6.] Obtain solution $\vw^{\star} = \vw^{(k-1)}$.
\end{itemize}

\section{Numerical Results}\label{sec:Numeric}
In this section, we analyze the performance of the proposed algorithms \ac{SAFP} and \ac{RMDI}, by considering the asymmetry of \ac{UL} and \ac{DL} traffic in two-cell scenario. The distance between the two \acp{BS} is $2$ km. The transmit power of \ac{BS} and \ac{UE} are $43$ and $22$ dBm respectively and all the other simulation parameters mainly related to channel gain can be found in \cite[Tab. A2.1.1-2]{3GPP36814}. We define the minimum time unit $\delta_t$ as $0.5$ ms and the minimum frequency unit $\delta_f$ as $15$ kHz. Further we have $W_t = 20$ and $W_f = 300$, i. e., a resource plane that spans a time duration of $0.01$ seconds and frequency of $5$ MHz (including the guard band).  

We defined a fixed total traffic demand $\Lambda=\sum_s d_s = 50$ kbits within $W_t\delta_t = 0.01$ seconds, which implies a total serving data rate of $5$ Mbit/s. The total traffic can be asymmetrically distributed between the two cells with different ratios among $\Ts_{\text{inter}} := \{1/9, 2/8, 3/7, \ldots, 9/1, 10/0\}$.  Within each cell, the traffic can be asymmetrically distributed between \ac{UL} and \ac{DL} traffic  with ratios among $\Ts_{\text{intra}}:=\{1/9, 2/8, 3/7, \ldots, 9/1\}$. \acp{UE} with either \ac{UL} or \ac{DL} traffic are generated with uniform distribution within the intersection of two balls with radius $2$ km, and with \ac{BS} 1 and 2 as their centers respectively, to analyze the scenario of high inter-cell interference. Without loss of generality, we can place one \ac{UL} and one \ac{DL} service in each cell with the traffic demand computed by the traffic ratio mentioned above.

{\it 1) Algorithm convergence of \ac{SAFP}.} 
Let us first examine the convergence of Algorithm \ref{algo:AlgorithmUpdateW}, and compare it with Algorithm \lq\lq FP\rq\rq \ that is summarized in Proposition \ref{prop:solutionFixedC} with conventional \ac{SINR} model \eqref{eqn:SINR_old}. The parameters are set as $N_{\text{max}}=30$, $N_{\text{iter}} =1000$, $\epsilon = 10^{-4}$. In Fig. \ref{fig:SAFP_convergence} we show the convergence of the \ac{SAFP} with one particular initialization of $\vw'$ and $\mC\left(\vw'\right)$ and compare it with FP. The magenta circle indicates the starting point with an updated $\mC\left(\vw^{(j)}\right)$, and the green dashed line shows that with each fixed $\mC\left(\vw^{(j)}\right)$, by performing fixed point iteration, $\rho$ monotonically increases and converges to the fixed point with respect to $\mC\left(\vw^{(j)}\right)$. 
%Note when updating $\mC\left(\vw^{(j)}\right)$, at the beginning $\rho$ may decrease since $\vw^{(j)}$ is updated based on the previous approximation $\mC\left(\vw^{(j-1)}\right)$. Also 
Note that the green dashed line is not the \lq\lq actual\rq\rq \ utility $\rho$, since it is computed with updated $\vw^{(i)}$ and the approximation $\mC\left(\vw^{(j-1)}\right)$. Therefore, we plot the red line to show the convergence of the actual utility at each step of updating $\mC$, computed with $\vw^{(j)}$ and $\mC\left(\vw^{(j)}\right)$. By comparing the red curve and the blue curve (convergence of FP algorithm), we observe a significant increase of utility $\rho$ by using \ac{SAFP}. This is because, comparing with FP that randomly places the \ac{UL} and \ac{DL} resource, \ac{SAFP} is based on an improved interference model, where \ac{ICI} only appears in the intersection of the sets of allocated \acp{MRU} between different cells. 
Fig. \ref{fig:SAFP_initialization} illustrates that with each random initialization of $\vw'$, the proposed \ac{SAFP} converges to a fixed point. The example shows that $30$ initializations converge to two different fixed points with utilities $4.35$ and $1.19$ respectively. $\vw^{\ast}$ corresponding to higher utility is chosen as the final solution. 
 
{\it 2) Performance comparison.} We compare the performance of \ac{SAFP} and \ac{RMDI} to the performance of the other three protocols, described in below.
\begin{itemize}
\item {\it FIX:} Fixed ratio and same position of the \ac{UL} and \ac{DL} resource in different cell. \ac{IMI} does not exist due to the orthogonal frequency band for  \ac{UL} and \ac{DL}. The amounts of the \ac{UL} and \ac{DL} resource are fixed to be the same.  
\item {\it dTDD:} Adaptive \ac{UL} and \ac{DL} resource proportional to the traffic volume in each cell independently. 
% However, all resources are allocated to either \ac{UL} or \ac{DL} traffic,   severe inter-mode interference may occur. 
\item {\it FP:} Proposed algorithm in \cite{liao2016dynamic} (summarized in Proposition \ref{prop:solutionFixedC}) that solves Problem \ref{prob:maxminQoS} with old \ac{SINR} model \eqref{eqn:SINR_old}. %The algorithm adapts joint \ac{UL} and \ac{DL} resource allocation to optimize the worst-case \ac{QoS}, without specifying the position of the \ac{UL} and \ac{DL} position. To reduce the inter-cell interference, the solution may allow partial resource to be blank in some cells, while at least one cell yields the maximum load $1$, implied by load constraint \eqref{eqn:loadconstraints}.
\end{itemize}
To compare the performance of protocols FIX, dTDD, FP, \ac{SAFP}, and \ac{RMDI} under different traffic asymmetry, we define a measure {\it inter-cell traffic distance},  given by  $D_{m,n}:= \|\ve{\vartheta}_n - \ve{\vartheta}_m\|$, where $\ve{\vartheta}_n:=\left[\vartheta_n^{\ul}, \vartheta_n^{\dl}\right]^T$ characterizes the \ac{UL} and \ac{DL} traffic distribution in cell $n$, and $\vartheta_n^{(x)}:=\left[\mB^{(x)}\vd\right]_n/\Lambda$, $n=1,2$, $x\in\{\text{u}, \text{d}\}$ denotes the fraction of the total traffic $\Lambda$ that traffic of type $x$ in cell $n$ accounts for, such that $\sum_{n\in\Ns} \sum_{\text{x}\in\{\text{u}, \text{d}\}}\vartheta_n^{(\text{x})} =1$. For example, if $\ve{\vartheta}_1 = \ve{\vartheta}_2 = [0.25, 0.25]^T$, we have $D_{1,2} =0$. 

Fig. \ref{fig:cdf_low} and \ref{fig:cdf_high} show the \ac{CDF} of utility $\rho$ derived by applying the five protocols under low and high inter-cell traffic distance, respectively.  The \ac{CDF} is derived from $1000$ simulation run times, each with different user locations and channel propagation, for every combination of the inter-cell traffic distribution ratio in set $\Ts_{\text{inter}}$ and intra-cell traffic distribution ratio in set $\Ts_{\text{intra}}$. All cases with $D_{1,2}\leq 0.5$ are considered as {\it low inter-cell traffic distance}, while with $D_{1,2}>0.5$ as {\it high inter-cell traffic distance}. 

Both Fig. \ref{fig:cdf_low} and \ref{fig:cdf_high} show that \ac{CDF} $F_d^{(\text{dTDD})}(1)>0.95$ for dTDD, implying that {\it service outage probability}, i.e., the probability that at least one service cannot be served with satisfied \ac{QoS} requirement,  is above $95$\%. The performance is worse than protocol FIX with $F_d^{(\text{FIX})}(1)>0.45$. This is because although \ac{UL}/\ac{DL} resource splitting is adapted to the traffic volume, the full occupation of the resource may cause severe \ac{IMI} to some services. Such observation encourages the application of our proposed algorithms, which are able to reduce the interference coupling among services. 
% and provide a good tradeoff between the resource utilization and interference reduction. 
By comparing FP, \ac{SAFP} and \ac{RMDI}, we show that FP further decreases the outage probability to below $20$\%, and \ac{SAFP} and \ac{RMDI} significantly outperform FP, with the outage probability for low traffic distance below $10$\%.
%, and that over $60$\% of the services are served with more than double of the required \ac{QoS}, i.e., $\rho>2$. 
Among the three, \ac{RMDI} provides the best performance of the utility distribution. By comparing Fig. \ref{fig:cdf_low} and \ref{fig:cdf_high}, we observe that \ac{SAFP} and \ac{RMDI} provides even higher performance gain under high traffic asymmetry.

{\it 3) Performance gain depending on traffic asymmetry.}
To analyze the performance gain depending on the traffic asymmetry, we average the utility obtained from $1000$ simulation run times for $D_{1,2}$ falling into the intervals $[0,0.16)$, $[0.16,0.32)$, $[0.32,0.48)$, $[0.48,0.64)$, $[0.64, 0.80)$, $[0.80, 1]$, respectively. Let us consider FIX as the baseline. Fig. \ref{fig:traffic_dist} shows that the performance of FIX decreases with the traffic asymmetry, and the average utility is below $1$ (infeasible \ac{QoS} target) when traffic distance $D_{1,2}>0.6$. Although dTDD adaptively splits the \ac{UL}/\ac{DL} resource, the full occupation of the resource causes severe \ac{IMI}, leading to the worst performance. On the other hand, FP reduces interference coupling among services, and provides $25$\% gain when traffic asymmetry is low, and almost $2$-fold gain when the asymmetry is ultra high. The proposed \ac{SAFP} incorporates interference coupling with \ac{UL}/\ac{DL} resource localization, which improves the gain to $2$-fold when the traffic asymmetry is low while $2.7$-fold when asymmetry is high. The enhanced version \ac{RMDI} further improves the gain by muting partial resource for interference cancellation. The gain is more significant when the traffic is highly asymmetric, achieving $3.2$-fold increase when $D_{1,2}\geq 0.64$.
\begin{figure}[t]
    \centering
		    \begin{subfigure}[t]{1\columnwidth}
        \centering
        \includegraphics[width=.8\columnwidth]{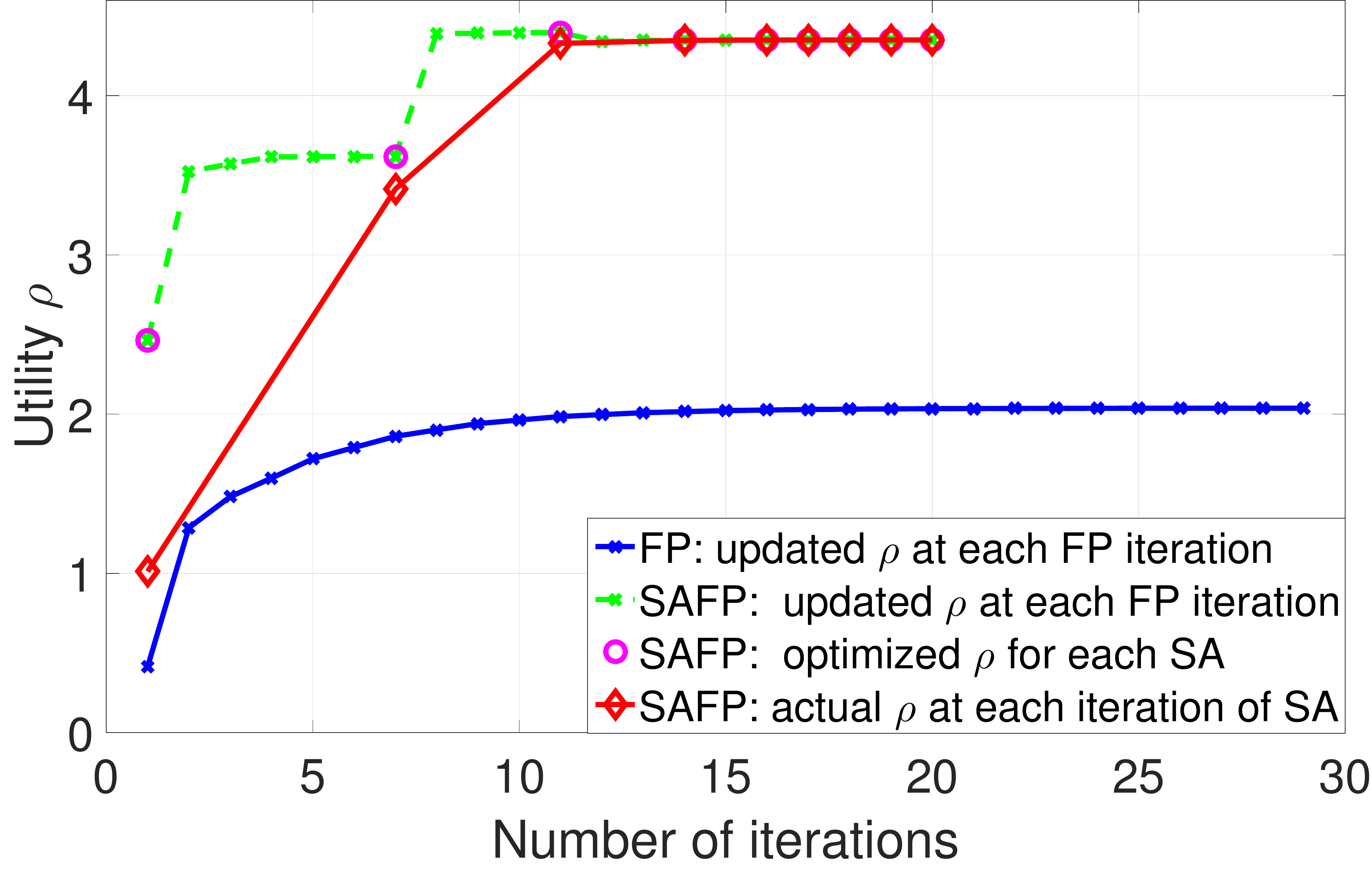}
        \caption{Comparison between FIX and SAFP.}
 \label{fig:SAFP_convergence}
    \end{subfigure}
    \begin{subfigure}[t]{1\columnwidth}
        \centering
        \includegraphics[width=.8\columnwidth]{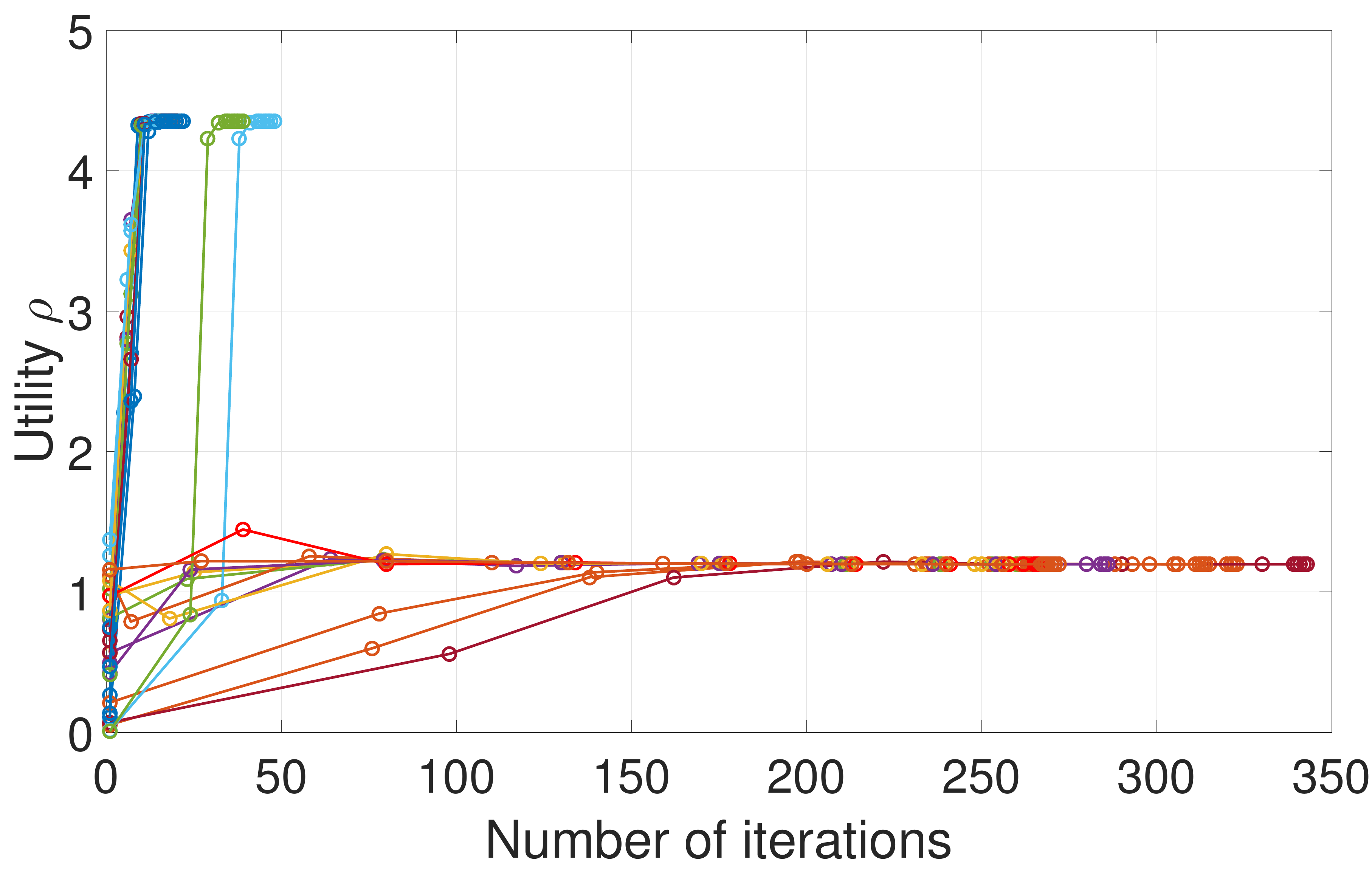}
        \caption{Examination of the random initialization. An example: With $30$ randomly initialized $\hat{\vw}$, SAFP converges to two local optima with $\rho^{\ast}(1) = 4.35$ and $\rho^{\ast}(2) = 1.19$.}
 \label{fig:SAFP_initialization}
    \end{subfigure}
    \caption{Examination of SAFP.}
		\label{fig:SAFP}
\end{figure}
\begin{figure}[t]
    \centering
		    \begin{subfigure}[t]{1\columnwidth}
        \centering
        \includegraphics[width=.8\columnwidth]{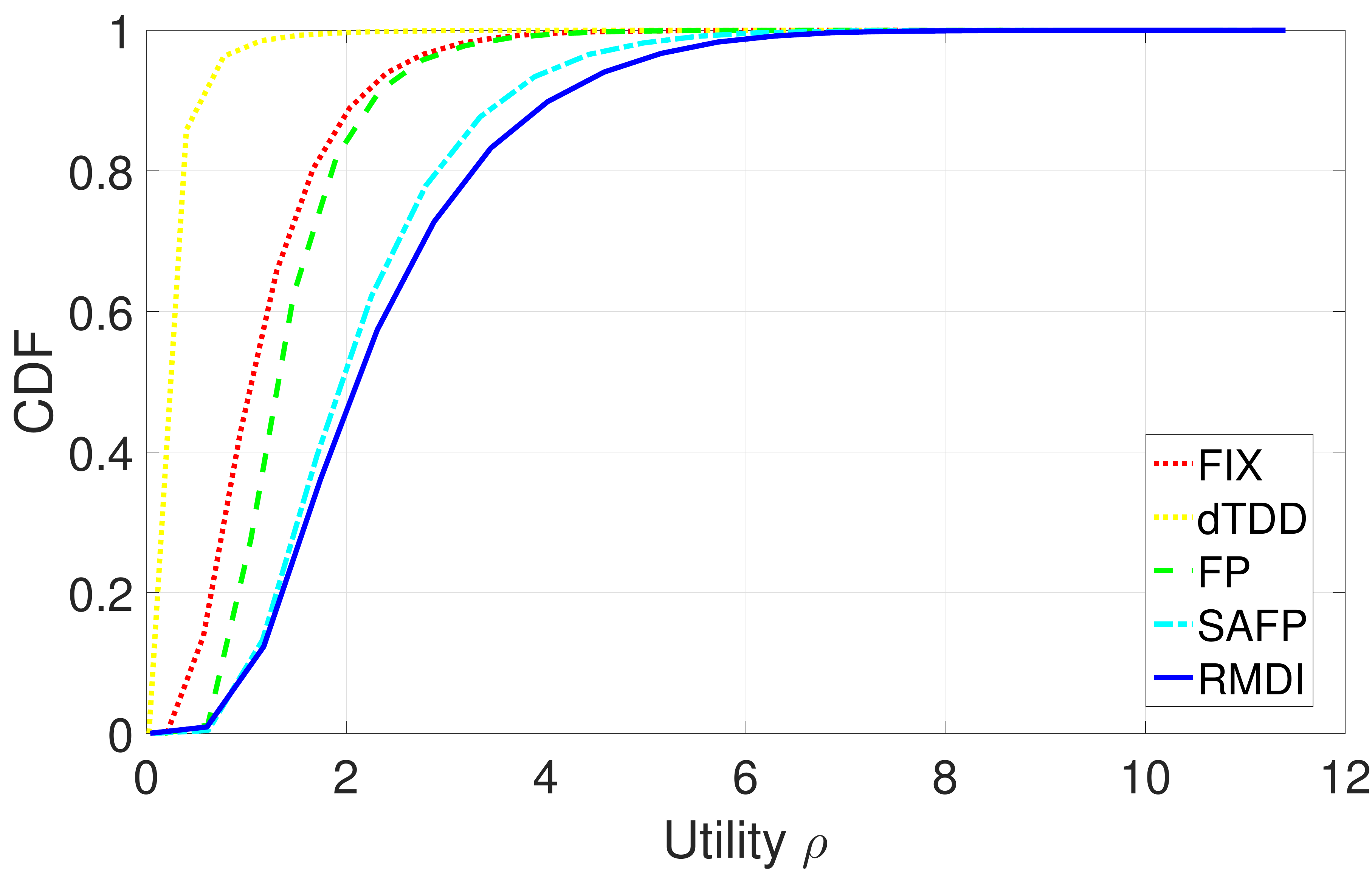}
        \caption{Utility \ac{CDF} under low inter-cell traffic distance.}
 \label{fig:cdf_low}
    \end{subfigure}
    \begin{subfigure}[t]{1\columnwidth}
        \centering
        \includegraphics[width=.8\columnwidth]{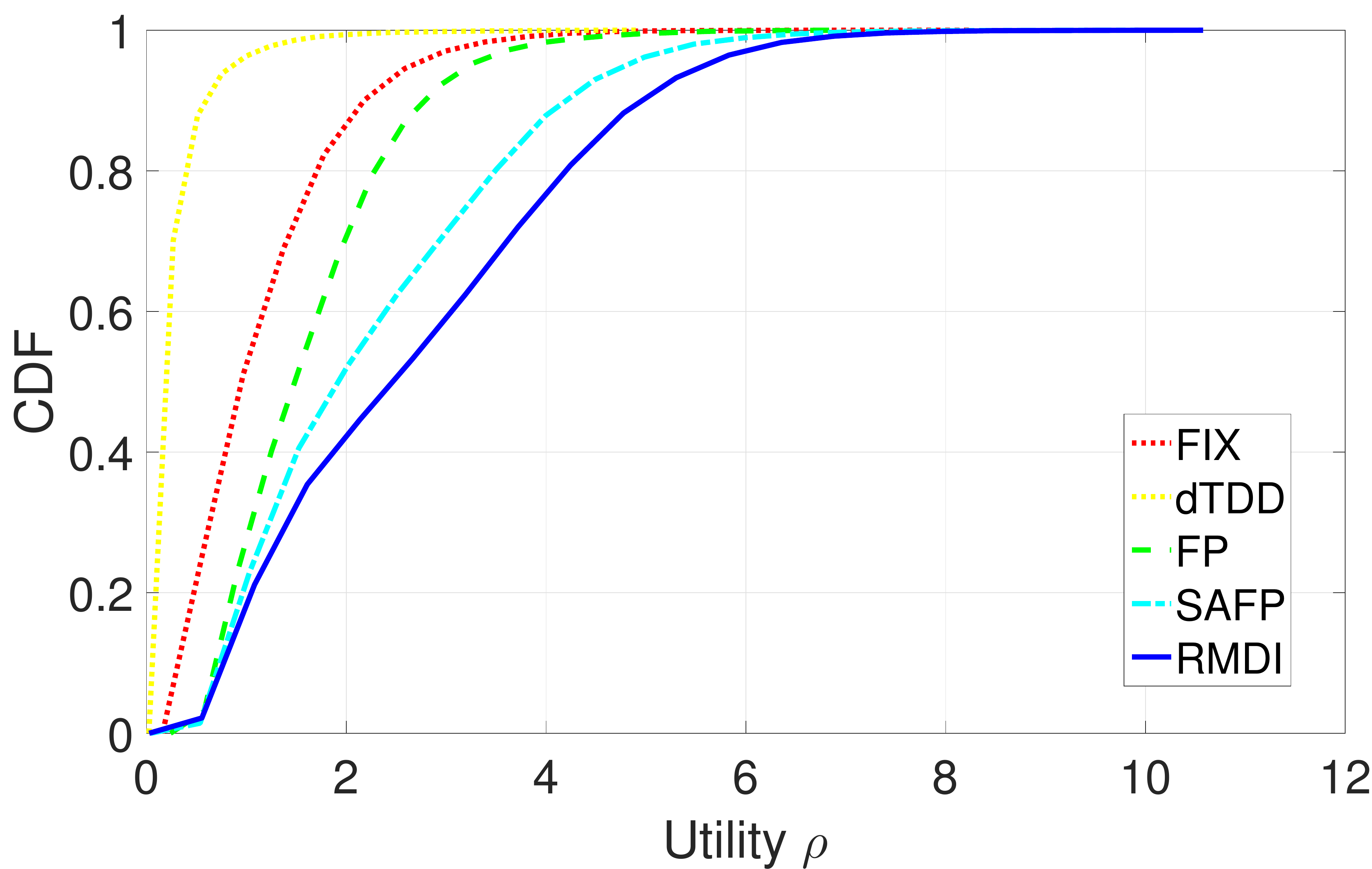}
        \caption{Utility \ac{CDF} under high inter-cell traffic distance.}
 \label{fig:cdf_high}
    \end{subfigure}
		  \begin{subfigure}[t]{1\columnwidth}
        \centering
        \includegraphics[width=.8\columnwidth]{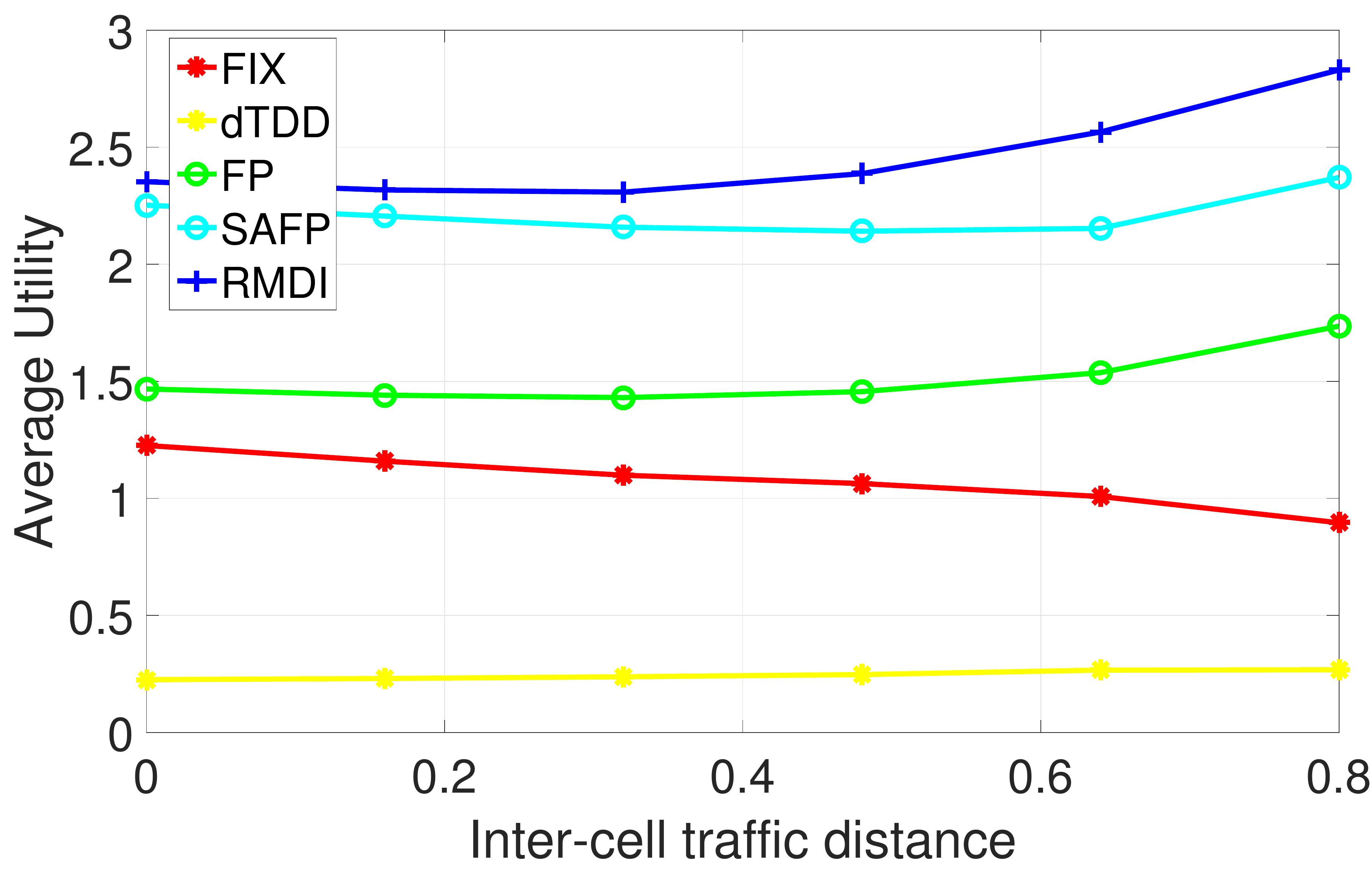}
        \caption{Average utility under different inter-cell traffic distance.}
 \label{fig:traffic_dist}
    \end{subfigure}
    \caption{Performance comparison among protocols.}
		\label{fig:Comparison}
\end{figure}
%
% \section{Conclusions} \label{sec:concl}

% define the heterogeneity between [0,1], where 1 most heterogeneity, 0 unified traffic
% \begin{equation}
% 1-\frac{1}{\log C}\sum_{i=1}^C q_i\log \frac{1}{q_i}
% \end{equation}
% where $q_i$ is the fraction of the load in $i$-th data flow (regardless of uplink or downlink). 
\appendices
\section{}\label{sec:app_Background}
\begin{definition}
A vector function $\vf:\RN^k\to \RP^k$ is a standard interference function (SIF) if the following axioms hold:
\begin{itemize}
\item[1.] (Monotonicity) $\vx \leq \vy$ implies $\vf(\vx)\leq\vf(\vy)$
\item[2.] (Scalability) for each $\alpha>1$, $\alpha\vf(\vx)>\vf(\alpha\vx)$ 
\end{itemize}
\label{def:SIF}
\vspace{-1ex}
\end{definition}
% The original definition of standard interference function is stated in, which also requires positivity. 
In Definition \ref{def:SIF} we drop positivity from its original definition \cite{Yates95a} because it is a consequence of the other two properties \cite{leung2004convergence}.

%\begin{theorem}[{\cite[Theorem 3.2]{nuzman2007contraction}}]
%Fix any monotonic norm $\|\cdot\|$ and \ac{SIF} function $\vf:\RN^k\to\RP^k$. 
%For each $\theta>0$ there is exactly one eigenvector $\vx'\in\RP^k$ and associate eigenvalue $\rho'$ of $\vf$ such that $\rho'\vx' = \vf(\vx')$ and $\|\vx'\|=\theta$. The repeated iteration
%\begin{equation}
%\vx^{(t+1)} =\frac{\theta \vf(\vx^{(t)})}{\|\vf(\vx^{(t)})\|}, \ t\in\NN,
%\label{eqn:FP_scaling_mapping}
%\end{equation}
%converges to the unique vector $\vx'$, which is called the fixed point of $\vf$. The associate eigenvalue is $\rho'=\|\vf(\vx')\|/\theta$.
%\label{Theo:MSS_FP}
%\end{theorem} 

\acrodef{3GPP}{3rd generation partnership project}
\acrodef{5G}{fifth generation}
% A
 \acrodef{ABS}{almost blank subframe}
% B
    \acrodef{BS}{base station}
% C
    \acrodef{CDF}{cumulative distribution function}
    \acrodef{CSI}{channel state information}
    \acrodef{CQI}{channel quality indicator}
% D
\acrodef{DL}{downlink}
    \acrodef{DUDe}{downlink and uplink decoupling}
% E
\acrodef{eICIC}{enhanced intercell interference coordination}
\acrodef{ESD}{energy spectral density}
% F
\acrodef{FDD}{frequency division duplex}
    \acrodef{FDMA}{frequency division multiple access}
% G
   \acrodef{GP}{Gaussian process}
    \acrodef{GPS}{global positioning system}
% H
\acrodef{HetNet}{heterogeneous network}
% I 
    \acrodef{ICI}{inter-cell interference}
		\acrodef{IMI}{inter-mode interference}
% J
% K
% L
\acrodef{LTE}{long term evolution}

% M
\acrodef{MAC}{media access control}
\acrodef{MRU}{minimum resource unit}
% N
% O
   \acrodef{OFDM}{orthogonal frequency division multiplexing}
% P
    \acrodef{PDF}{probability density function}
    \acrodef{PHY}{physical layer}
		\acrodef{PSD}{power spectral density}
    \acrodef{PRB}{physical resource block}
% Q
   \acrodef{QoE}{quality of experience}
    \acrodef{QoS}{quality of service}
% R
    \acrodef{RAN}{radio access network}
		\acrodef{RBS}{removal of bottleneck services}
		\acrodef{RMDI}{resource muting for dominant interferer}
    \acrodef{RRM}{radio resource management}
		\acrodef{RU}{resource unit}
		\acrodef{RX}{receiver}
% S
 \acrodef{SAFP}{successive approximation of fixed point}
    \acrodef{SDN}{software defined network}
    \acrodef{SNR}{signal-to-noise ratio}
    \acrodef{SINR}{signal-to-interference-plus-noise ratio}
\acrodef{SIR}{signal-to-interference ratio}
\acrodef{SIF}{standard interference function}
    \acrodef{SVM}{support vector machine}
% T
    \acrodef{TCP}{transmission control protocol}
		\acrodef{TDD}{time division duplex}
    \acrodef{TDMA}{time division multiple access}
		\acrodef{TTI}{transmission time interval}
		\acrodef{TX}{transmitter}
		% U 
		\acrodef{UE}{user equipment}
		\acrodef{UL}{uplink}
		% V
		% W
    \acrodef{WLAN}{wireless local area network}
% X
% Y 
% Z

\bibliographystyle{IEEEtran}
\bibliography{refs}
		
\end{document}